\documentclass[a4paper,11pt]{article}

\usepackage{amsmath}
\usepackage{amssymb}
\usepackage{amsthm}
\usepackage{graphicx}
\usepackage{enumerate}
\usepackage{cite}
\usepackage{fullpage}
\usepackage[usenames]{color}
\usepackage[nodayofweek]{datetime}
\usepackage{hyperref}

\allowdisplaybreaks

\newcommand{\conv}{\ensuremath{\mathrm{conv}}}
\newcommand{\Real}{\ensuremath{\mathbb{R}}}
\newcommand{\Plane}{\ensuremath{\mathbb{R}^2}}
\newcommand{\DPlane}{\ensuremath{\mathbb{D}}}

\newcommand{\bd}{\ensuremath{\partial}}

\newcommand{\intr}{\ensuremath{\mathrm{int}}}

\newcommand{\seg}{\overline}

\newcommand{\dual}[1]{\ensuremath{{#1}^\star}}
\newcommand{\arr}{\ensuremath{\mathcal{A}}}
\newcommand{\darr}{\ensuremath{\dual{\mathcal{A}}}}
\newcommand{\uenv}{\ensuremath{\mathcal{U}}}
\newcommand{\lenv}{\ensuremath{\mathcal{L}}}
\newcommand{\mline}{\ensuremath{\mathcal{M}}}
\newcommand{\qabove}{\ensuremath{\mathcal{Q}^+}}
\newcommand{\qbelow}{\ensuremath{\mathcal{Q}^-}}
\newcommand{\qbetween}{\ensuremath{\mathcal{Q}}}
\newcommand{\tmap}{\ensuremath{\mathcal{T}}}
\newcommand{\extr}{\ensuremath{\chi}}

\newtheoremstyle{mytheorem}{3pt}{3pt}{\slshape}{}{\bfseries}{}{.5em}{}
\theoremstyle{mytheorem}
\newtheorem{lemma}{Lemma}
\newtheorem{theorem}{Theorem}

\newtheorem{observation}{Observation}

\theoremstyle{definition}

\makeatletter
\newbox\ProofSym
\setbox\ProofSym=\hbox{%
\unitlength=0.18ex%
\begin{picture}(10,10)
\put(0,0){\framebox(9,9){}} \put(0,3){\framebox(6,6){}}
\end{picture}}
\renewenvironment{proof}[1][Proof.]{\O@proof{#1}}{\O@endproof}
\def\O@proof#1{\trivlist
   \@topsep\z@\@topsepadd\smallskipamount%
   \@ifstar{\item[]}{\item[\hskip\labelsep\it #1 ]}}
\def\O@endproof{\hfill\copy\ProofSym\linebreak[3mm]\endtrivlist}
\makeatother

\def\denseitems{
    \itemsep1pt plus1pt minus1pt
    \parsep0pt plus0pt
    \parskip0pt\topsep0pt}
\begin{document}


\title{Minimum-Width Double-Strip and Parallelogram Annulus%
\thanks{%
This work was supported by Kyonggi University Research Grant 2018.
}
}

\author{%
Sang Won Bae\footnote{%
Division of Computer Science and Engineering, Kyonggi University, Suwon, Korea.
Email: \texttt{swbae@kgu.ac.kr} }
}

\date{%
\today\quad\currenttime
}

\maketitle

\begin{abstract}
In this paper, we study the problem of computing a minimum-width
double-strip or parallelogram annulus
that encloses a given set of $n$ points in the plane.
A double-strip is a closed region in the plane whose boundary consists of
four parallel lines and a parallelogram annulus is a closed region
between two edge-parallel parallelograms.
We present several first algorithms for these problems.
Among them are $O(n^2)$ and $O(n^3 \log n)$-time algorithms
that compute a minimum-width double-strip and parallelogram annulus,
respectively, when their orientations can be freely chosen.\\

\end{abstract}

\section{Introduction} \label{sec:intro}

The \emph{minimum-width annulus problem} asks to find an annulus
of a certain shape with the minimum width that encloses a given set $P$
of $n$ points in the plane.
An annulus informally depicts a ring-shaped region in the plane.
As the most natural and classical example, a circular annulus is
defined to be the region between two concentric circles.
If one wants to find a circle that best fits an input point set $P$,
then her problem can be solved by finding out a minimum-width
circular annulus that encloses $P$.
After early results on the circular annulus problem~\cite{rz-epccmrsare-92},
the currently best algorithm that computes a minimum-width circular annulus
that encloses $n$ input points takes $O(n^{\frac{3}{2}+\epsilon})$ time~\cite{ast-apsgo-94,as-erasgop-96} for any $\epsilon > 0$.
Analogously, such a problem of matching a point set into a closed curve class
can be formulated into the minimum-width annulus problem for annuli of
different shapes.

Along with applications not only to the points-to-curve matching problem
but also to other types of facility location,
the minimum-width annulus problem has been extensively studied
for recent years,
with a variety of variations and extensions.
Abellanas et al.~\cite{ahimpr-bfr-03} considered minimum-width rectangular annuli
that are axis-parallel, and presented two algorithms taking $O(n)$ or $O(n \log n)$ time:
one minimizes the width over rectangular annuli with arbitrary aspect ratio
and the other does over  rectangular annuli with a prescribed aspect ratio, respectively.
Gluchshenko et al.~\cite{ght-oafepramw-09} presented an $O(n \log n)$-time algorithm
that computes a minimum-width axis-parallel square annulus,
and proved a matching lower bound,
while the second algorithm by Abellanas et al.\@ can do the same in the same time bound.
If one considers rectangular or square annuli in arbitrary orientation,
the problem becomes more difficult.
Mukherjee et al.~\cite{mmkd-mwra-13} presented an $O(n^2 \log n)$-time algorithm
that computes a minimum-width rectangular annulus in arbitrary orientation
and arbitrary aspect ratio.
The author~\cite{b-cmwsaao-18} showed that
a minimum-width square annulus in arbitrary orientation can be computed
in $O(n^3 \log n)$ time, and recently improved it to $O(n^3)$ time~\cite{b-marsap-19X}.

In this paper, we consider a more generalized shape, namely, parallelograms, and
annuli based on them in fixed or arbitrary orientation,
which have at least one more degree of freedom than
square and rectangular annuli have.
More precisely, we define a parallelogram annulus as a closed region
between two edge-parallel parallelograms, and
address the problem of computing a minimum-width parallelogram annulus
that encloses the input point set $P$.
(See \figurename~\ref{fig:strip_annulus}(c).)
We consider several restricted cases of the problem about
two orientations of sides of parallelogram annuli.
Our main results are summarized as follows:
\begin{enumerate}[(1)] \denseitems
 \item When both orientations for sides are fixed,
 a minimum-width parallelogram annulus that encloses $P$
 can be computed in $O(n)$ time.
 \item When one orientation for sides is fixed and the other can be
 chosen arbitrarily,
 a minimum-width parallelogram annulus that encloses $P$
 can be computed in $O(n^2)$ time.
 \item A minimum-width parallelogram annulus that encloses $P$
 over all pairs of orientations can be computed in $O(n^3\log n)$ time.
\end{enumerate}

To obtain these algorithms for the problem,
we also introduce another geometric optimization problem,
called the \emph{minimum-width double-strip problem},
which asks to compute a double-strip of minimum width
that encloses $P$.
A double-strip is defined to be the union of two parallel strips,
where a strip is the region between two parallel lines in the plane.
(See \figurename~\ref{fig:strip_annulus}(b).)
We show that this new problem is closely related to the parallelogram annulus problem.
The minimum-width double-strip problem has its own interest
as a special case of the \emph{two-line center problem},
in which one wants to find two strips, possibly being non-parallel,
that encloses $P$ and minimizes the width of the wider strip.
After the first sub-cubic algorithm is presented by
Agarwal and Sharir~\cite{as-pglp-94},
the currently best algorithm for the two-line center problem
takes $O(n^2 \log^2 n)$ time~\cite{jk-tlcppv:nads-95,gks-sgsopsm-98}.

To our best knowledge, however, no nontrivial result on the double-strip problem
is known in the literature.
In this paper, we obtain the following algorithmic results:
\begin{enumerate}[(1)] \denseitems
 \setcounter{enumi}{3}
 \item A minimum-width double-strip that encloses $P$ over all orientations
 can be computed in $O(n^2)$ time.
 \item We also consider a constrained version of the problem
 in which a subset $Q\subseteq P$ with $k = |Q|$ is given and one wants to find
 a minimum-width double-strip enclosing $Q$ such that
 all points of $P$ should lie in between its outer boundary lines.
 We show that this can be solved in $O(n \log n + kn)$ time.
 \item We further address some online and offline versions of
 the dynamic constrained double-strip problem under insertions and/or deletions of a point
 on the subset $Q$ to enclose.
\end{enumerate}

The rest of the paper is organized as follows:
Section~\ref{sec:pre} introduces necessary definitions and preliminaries.
Section~\ref{sec:double-strip} is devoted to solve
the minimum-width double-strip problem and present an $O(n^2)$ time algorithm,
which is generalized to the constrained double-strip problem
in Section~\ref{sec:c_double-strip}.
The minimum-width parallelogram annulus problem is finally discussed and solved
in Section~\ref{sec:paralannul}.

\section{Preliminaries} \label{sec:pre}

In this section, we introduce definitions of necessary concepts
and preliminaries for further discussion.
For any subset $A \subseteq \Plane$ of the plane $\Plane$,
its boundary and interior are denoted by
$\bd A$ and $\intr A$, respectively.

\begin{figure}[tb]
\begin{center}
\includegraphics[width=.95\textwidth]{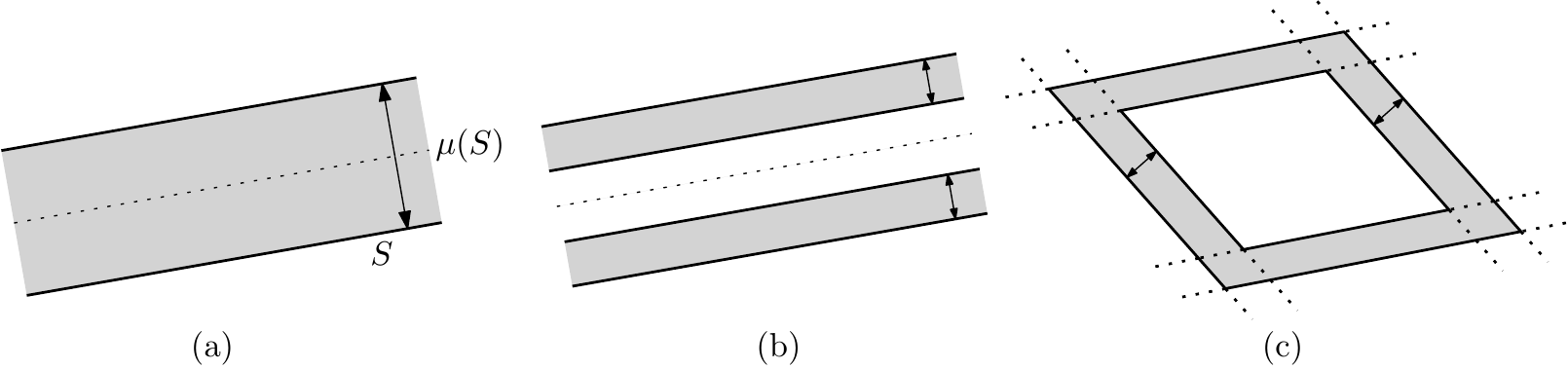}
\end{center}
\caption{Illustrations to (a) a strip $S$ and its middle line $\mu(S)$,
 (b) a double-strip, and (c) a parallelogram annulus.
 The arrows depict the width of each shape.
 }
\label{fig:strip_annulus}
\end{figure}

\paragraph*{Strips and double-strips.}
For two parallel lines $\ell$ and $\ell'$ in the plane $\Plane$,
the \emph{distance} between $\ell$ and $\ell'$
denotes the length of any line segment
that is orthogonal to $\ell$ and $\ell'$
and have endpoints one on $\ell$ and the other on $\ell'$.
A \emph{strip} is a closed region bounded by two parallel lines in the plane.
For any strip $S$, its \emph{width} $w(S)$ is the distance
between its bounding lines,
and its \emph{middle line} $\mu(S)$ is the line parallel to its bounding lines
such that the distance between $\mu(S)$ and each of the bounding lines
is exactly half the width $w(S)$ of $S$.
See \figurename~\ref{fig:strip_annulus}(a).

A \emph{double-strip} is the union of two disjoint parallel strips of equal width,
or equivalently, is a closed region obtained by a strip $S$
subtracted by the interior of another strip $S'$
such that $\mu(S) = \mu(S')$ and $S' \subseteq S$.
For any double-strip defined by two strips $S$ and $S'$ in this way,
$S$ is called its \emph{outer strip} and $S'$ its \emph{inner strip}.
The \emph{width} of such a double-strip $D$, denoted by $w(D)$, is defined to be
half the difference of the widths of $S$ and $S'$, that is,
$w(D) = (w(S) - w(S'))/2$.
See \figurename~\ref{fig:strip_annulus}(b).

\paragraph*{Parallelogram annuli.}
A parallelogram is a quadrilateral that is the intersection
of two non-parallel strips.
We define a parallelogram annulus to be a parallelogram $R$
with a parallelogram hole $R'$,
analogously as a circular annulus is a circle with a circular hole.
Here, we add a condition that the outer and inner parallelograms $R$ and $R'$
should be side-wise parallel.
There are several ways to define such a parallelogram annulus,
among which we introduce the following definition.
A \emph{parallelogram annulus} $A$ is defined by two double-strips $D_1$ and $D_2$
as follows:
\begin{enumerate} \denseitems
 \item The \emph{outer parallelogram} $R$ of $A$ is
 the intersection of the outer strips of $D_1$ and $D_2$.
 \item The \emph{inner parallelogram} $R'$ of $A$ is
 the intersection of the inner strips of $D_1$ and $D_2$.
 \item The parallelogram annulus $A$ is the closed region between $R$ and $R'$,
 that is, $A = R \setminus \intr R'$.
 \item The \emph{width} of $A$, denoted by $w(A)$, is taken to be the bigger one
 between the widths of $D_1$ and $D_2$, that is, $w(A) = \max\{w(D_1), w(D_2)\}$.
\end{enumerate}
See \figurename~\ref{fig:strip_annulus}(c) for an illustration.

The main purpose of this paper is to solve the
\emph{minimum-width parallelogram annulus problem} in which
we are given a set $P$ of points in the plane
and want to find a parallelogram annulus of minimum width that encloses $P$.
As discussed above, a parallelogram annulus is closely related to
strips and double-strips.
The \emph{minimum-width double-strip problem}
asks to find a double-strip of minimum width that encloses $P$
in fixed or arbitrary orientation.

\paragraph*{Orientations and the width function.}
The \emph{orientation} of a line or line segment $\ell$ in the plane is
a value $\theta \in [-\pi/2, \pi/2)$\footnote{%
The orientation $\theta$ is indeed of period $\pi$.
In this paper we choose $[-\pi/2, \pi/2)$ for the orientation domain.}
such that the rotated copy of the $x$-axis by $\theta$ counter-clockwise is parallel to $\ell$.
If the orientation of a line or line segment is $\theta$,
then we say that the line or line segment is \emph{$\theta$-aligned}.
A strip or a double-strip is also called $\theta$-aligned for some $\theta \in [-\pi/2, \pi/2)$
if its bounding lines are $\theta$-aligned.
A parallelogram or a parallelogram annulus is
\emph{$(\theta, \phi)$-aligned} if it is defined by
two double-strips that are $\theta$-aligned and $\phi$-aligned,
respectively.

\begin{figure}[tb]
\begin{center}
\includegraphics[width=.90\textwidth]{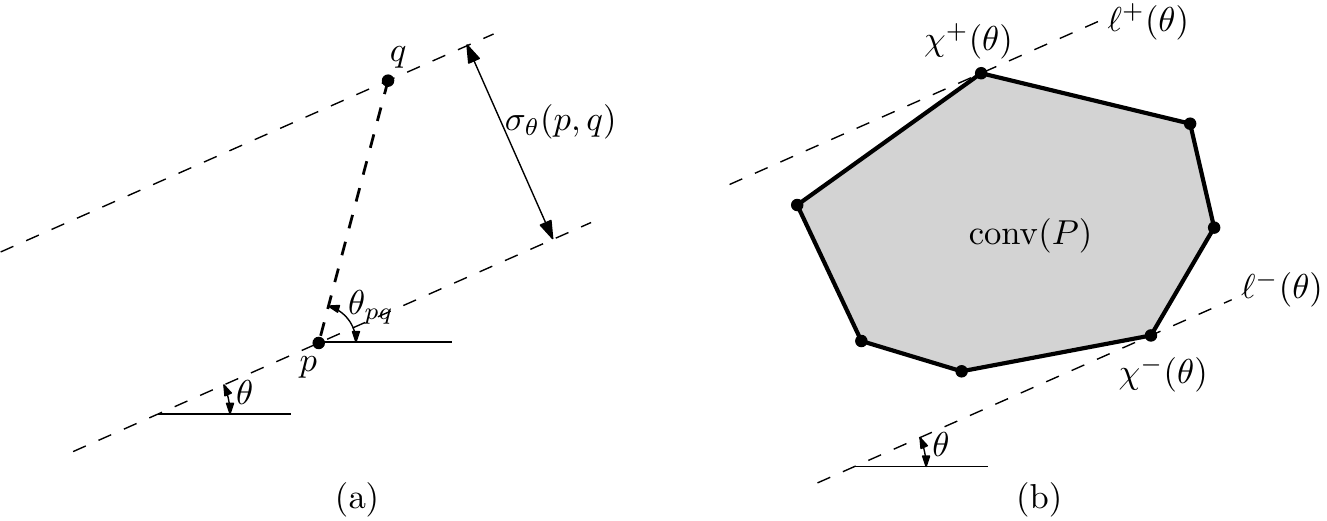}
\end{center}
\caption{(a) The $\theta$-aligned strip defined by $p$ and $q$, and its width
 $\sigma_\theta(p,q)$.
 (b) The antipodal pair $(\extr^+(\theta), \extr^-(\theta))$ of $P$ for
 orientation $\theta \in [-\pi/2, \pi/2)$.
 }
\label{fig:width_extreme}
\end{figure}

For any two points $p, q\in \Plane$, let $\seg{pq}$ denote the line segment joining $p$ and $q$,
and $|\seg{pq}|$ denote the Euclidean length of $\seg{pq}$.
We will often discuss the strip defined by two parallel lines
through $p$ and $q$, and its width.
Let $\sigma_\theta(p,q)$ be the width of the strip defined by
two $\theta$-aligned lines through $p$ and $q$, respectively.
It is not difficult to see that
 \[\sigma_\theta(p, q) = |\seg{pq}| \cdot | \sin (\theta - \theta_{pq})|, \]
where $\theta_{pq} \in [-\pi/2, \pi/2)$ denotes the orientation of $\seg{pq}$.
See \figurename~\ref{fig:width_extreme}(a).

A single-variate function of a particular form $a \sin(\theta + b)$
for some constants $a, b\in \Real$ with $a\neq 0$ is called
\emph{sinusoidal function (of period $2\pi$)}.
Obviously, the equation $a \sin(\theta + b) = 0$ has at most one zero over $\theta \in [-\pi/2, \pi/2)$.
The following property of sinusoidal functions is well known and easily derived.
\begin{lemma} \label{lemma:sinusoidal}
 The sum of two sinusoidal functions is also sinusoidal.
 Therefore, the graphs of two sinusoidal functions cross at most once over
 $[-\pi/2, \pi/2)$.
\end{lemma}
Note that, taking $\theta \in [-\pi/2,\pi/2)$ as a variable,
the function $\sigma_\theta(p,q)$ for fixed points $p,q \in \Plane$
is piecewise sinusoidal with at most one breakpoint.

\paragraph*{Extreme points and antipodal pairs.}
Let $P$ be a finite set of points in $\Plane$, and
let $\conv(P)$ be its convex hull.
The corners of $\conv(P)$ are called \emph{extreme} of $P$.
For each $\theta \in [-\pi/2, \pi/2)$,
let $S(\theta)$ be the minimum-width $\theta$-aligned strip enclosing $P$.
Then, the two lines that define $S(\theta)$ pass through two extreme points
of $P$, one on each.
More precisely, let $\ell^+(\theta)$ and $\ell^-(\theta)$ be the two lines
defining $S(\theta)$ such that $\ell^+(\theta)$ lies above $\ell^-(\theta)$
if $\theta \neq -\pi/2$, or $\ell^+(\theta)$ lies to the right of $\ell^-(\theta)$
if $\theta = -\pi/2$.
There must be an extreme point of $P$ on each of
the two lines $\ell^+(\theta)$ and $\ell^-(\theta)$,
denoted by $\extr^+(\theta)$ and $\extr^-(\theta)$.
If there are two or more such points of $P$, then
we choose the last one in the counter-clockwise order along the boundary of
$\conv(P)$.
The width of $S(\theta)$ is equal to
$\sigma_\theta(\extr^+(\theta), \extr^-(\theta))$.

For each $\theta \in [-\pi/2, \pi/2)$,
the pair $(\extr^+(\theta), \extr^-(\theta))$ is called \emph{antipodal}.
It is known that there are at most $O(n)$ different antipodal pairs
by Toussaint~\cite{t-sgprc-83}.
See \figurename~\ref{fig:width_extreme}(b).
Starting from $\theta = -\pi/2$, imagine the motion of
the two lines $\ell^+(\theta)$ and $\ell^-(\theta)$
as $\theta$ continuously increases.
Then, the antipodal pair $(\extr^+(\theta), \extr^-(\theta))$ for $\theta$
only changes when one of the two lines contains an edge of $\conv(P)$.
In this way, the orientation domain $[-\pi/2, \pi/2)$ is decomposed into
maximal intervals $I$ such that
$(\extr^+(\theta), \extr^-(\theta))$ remains the same in $I$.

\section{Minimum-Width Double-Strips} \label{sec:double-strip}
In this section, we address the problem of computing a minimum-width
double-strip that encloses a given set $P$ of $n$ points in $\Plane$.

We start with the problem in a given orientation $\theta \in [-\pi/2, \pi/2)$.
Let $w(\theta)$ be the minimum possible width of a $\theta$-aligned double-strip
enclosing $P$.
The following observation can be obtained by a simple geometric argument.
\begin{observation} \label{obs:double-strip-conf}
 For each $\theta\in [-\pi/2, \pi/2)$,
 there exists a minimum-width $\theta$-aligned double-strip $D(\theta)$
 enclosing $P$ whose outer strip is $S(\theta)$ and
 inner strip is $S'(\theta)$, where
 \begin{itemize} \denseitems
  \item $S(\theta)$ is the minimum-width $\theta$-aligned strip enclosing $P$, and
  \item $S'(\theta)$ is the maximum-width $\theta$-aligned strip
  such that its interior is empty of any point in $P$ and
  $\mu(S'(\theta)) = \mu(S(\theta))$.
 \end{itemize}
\end{observation}
\begin{proof}
Let $D$ be a minimum-width $\theta$-aligned double-strip enclosing $P$.
Let $S_1$ and $S_2$ be the two strips such that $D = S_1 \cup S_2$.
If the boundary of $S_1$ does not contain an extreme point of $P$,
then we can slide $S_1$ inwards until its boundary hits an extreme point of $P$.
Let $S'_1$ be the resulting strip after this sliding process.
Then, we have $S_1 \cap P \subset S'_1$.
In the same way, we slide $S_2$ until the boundary of $S_2$ hits
an extreme point of $P$, and let $S'_2$ be the resulting strip.
Then, it holds that $S_2 \cap P \subset S'_2$.
As a result, $P \subset (S'_1 \cup S'_2)$ since $P\subset (S_1 \cup S_2)$.
Note that the outer strip of the double-strip $D' := S'_1 \cup S'_2$
is exactly $S(\theta)$,
the minimum-width $\theta$-aligned strip enclosing $P$.

Now, let $S'$ be the inner strip of $D'$.
By definition, we have $\mu(S(\theta)) = \mu(S')$.
Suppose that the inner strip $S'$ of $D'$ is not equal to $S'(\theta)$,
the maximum-width $\theta$-aligned strip
such that its interior is empty of any point in $P$ and $\mu(S'(\theta)) = \mu(S(\theta))$.
Then, the boundary of $S'$ does not contain any point of $P$.
Hence, the width of $S'(\theta)$ is strictly larger than the width of $S'$,
a contradiction that $D'$ is of minimum width.
Therefore, the inner strip of $D'$ should be $S'(\theta)$.
\end{proof}

We focus on finding the minimum-width double-strip $D(\theta)$ described in
Observation~\ref{obs:double-strip-conf}.
The outer strip $S(\theta)$ of $D(\theta)$ is determined by
$\ell^+(\theta)$ and $\ell^-(\theta)$ on which the two extreme points
$\extr^+(\theta)$ and $\extr^-(\theta)$ lie.
For $p\in P$, let
$d_p(\theta):=\min\{\sigma_\theta(p, \extr^+(\theta)), \sigma_\theta(p, \extr^-(\theta))\}$.
Then, the width $w(\theta)$ of $D(\theta)$
in orientation $\theta$ is determined by
\[ w(\theta) = \max_{p \in P} d_p(\theta).\]

It is not difficult to see that $w(\theta)$ can be evaluated
in $O(n)$ time for a given $\theta \in [-\pi/2,\pi/2)$.
\begin{theorem} \label{thm:double-strip-fixed-orientation}
 Given a set $P$ of $n$ points and an orientation $\theta \in [-\pi/2, \pi/2)$,
 a minimum-width $\theta$-aligned double-strip enclosing $P$
 and its width $w(\theta)$ can be computed in $O(n)$ time.
\end{theorem}
\begin{proof}
In this proof, we describe our algorithm that computes $D(\theta)$.
Given $\theta$, one can compute the two extreme points
$\extr^+(\theta)$ and $\extr^-(\theta)$ in $O(n)$ time
by computing the minimum and maximum among the inner products of
\[ (\cos(\theta+\pi/2)), \sin(\theta+\pi/2)) \]
and all points in $P$ as vectors.
The antipodal pair also determines the outer strip $S(\theta)$.

After identifying $\extr^+(\theta)$ and $\extr^-(\theta)$,
the value of $d_p(\theta)$ for each $p\in P$ can be computed in $O(1)$ time.
Hence, $w(\theta) = \max_{p \in P} d_p(\theta)$ can be found
in additional $O(n)$ time.
This determines the inner strip $S'(\theta)$.
\end{proof}

Next, we turn to finding a minimum-width double-strip over all orientations.
This is equivalent to computing the minimum value of $w(\theta)$
over $\theta \in [-\pi/2, \pi/2)$, denoted by $w^*$.
Let $\theta^*$ be an optimal orientation such that $w(\theta^*) = w^*$.
Consider the corresponding double-strip $D(\theta^*)$
whose outer strip is $S(\theta^*)$ and inner strip is $S'(\theta^*)$,
as described in Observation~\ref{obs:double-strip-conf}.
We then observe the following for
the minimum-width double-strip $D(\theta^*)$ enclosing $P$.
\begin{lemma} \label{lem:double-strip-conf2}
 Let $\theta^*$ be an orientation such that $w(\theta^*) = w^*$.
 Then, either
 \begin{enumerate}[(a)] \denseitems
  \item three extreme points of $P$ lie on the boundary of $S(\theta^*)$, or
  \item two points of $P$ lie on the boundary of $S'(\theta^*)$.
 \end{enumerate}
\end{lemma}
\begin{proof}
From Observation~\ref{obs:double-strip-conf},
we know that the boundary of $S(\theta^*)$ contains at least two points,
$\extr^+(\theta^*)$ and $\extr^-(\theta^*)$,
while the boundary of $S'(\theta^*)$ contains at least one point $q\in P$
that maximizes $d_p(\theta^*)$ over $p\in P$.

Suppose that there is no more point in $P$ lying on the boundary
of $S(\theta^*)$ or of $S'(\theta^*)$.
Without loss of generality, we assume that
$q$ lies above $\mu(S(\theta))$.
This implies that the width of $D(\theta^*)$ is
 \[ w^* = w(\theta^*) = d_q(\theta^*) = \sigma_{\theta^*}(q, \extr^+(\theta^*)).\]
Consider the function $d_q(\theta)$ of $\theta$ near $\theta^*$.
Note that $d_q(\theta^*) = \sigma_{\theta^*}(q, \extr^+(\theta^*)$ is
a sinusoidal function.

Since there is no more extreme point on the boundary of $S(\theta^*)$,
there is $\theta'$ sufficiently close to $\theta^*$
such that $d_q(\theta') < d_q(\theta^*)$.
In addition, since there is no more point on the boundary of $S'(\theta^*)$,
we indeed have a strict inequality $d_p(\theta^*) < d_q(\theta^*)$
and hence $d_p(\theta') < d_q(\theta')$ for all $p\in P \setminus \{q\}$.
This implies that
 \[ w(\theta') = d_q(\theta') < d_q(\theta^*) = w(\theta^*) = w^*,\]
a contradiction.
Hence, there must be one more point in $P$ on the boundary of $S(\theta^*)$
or of $S'(\theta^*)$,
implying the lemma.
\end{proof}

Regard $d_p(\theta)$ as a function of $\theta$ in domain $[-\pi/2, \pi/2)$.
This function depends on the antipodal pair of extreme points
$(\extr^+(\theta), \extr^-(\theta))$ for $\theta$.
Since there are $O(n)$ antipodal pairs~\cite{t-sgprc-83},
the function $d_p$ for each $p\in P$
is piecewise sinusoidal with $O(n)$ breakpoints.
The width function $w(\theta)$ is the upper envelope
of the $n$ functions $d_p(\theta)$ for $p\in P$,
consisting of $O(n^2)$ sinusoidal curves in total.
Thus, we can compute the minimum width
$w^* = \min_{\theta \in [-\pi/2, \pi/2)} w(\theta)$
by computing the upper envelope of these $O(n^2)$ sinusoidal curves
and finding a lowest point in the envelope.
Since any two sinusoidal curves cross at most once by Lemma~\ref{lemma:sinusoidal},
this can be done in $O(n^2 \log n)$ time~\cite{h-fuenls-89}.
The width function $w$ is piecewise sinusoidal with $O(n^2 \alpha(n))$ breakpoints, where $\alpha$ denotes the inverse Ackermann function.
and its lowest point always occurs at one of the breakpoints
by Lemma~\ref{lem:double-strip-conf2}.
Hence, $O(n^2 \log n)$ time is sufficient to solve the problem.

In the following, we improve this to $O(n^2)$ time.
As observed above, the double-strip $D(\theta)$ of width $w(\theta)$
is determined by its outer strip $S(\theta)$ and inner strip $S'(\theta)$.
Let $\mu(\theta) := \mu(S(\theta))=\mu(S'(\theta))$ be the middle line of $D(\theta)$.
The middle line $\mu(\theta)$ separates $P$
into two subsets $P^+(\theta)$ and $P^-(\theta)$,
where $P^+(\theta)$ is the set of points $p\in P$ lying in the strip
defined by $\ell^+(\theta)$ and $\mu(\theta)$,
and $P^-(\theta) = P\setminus P^+(\theta)$.
Define $q^+(\theta) \in P^+(\theta)$ to be the closest point to line $\mu(\theta)$
among $P^+(\theta)$, and similarly $q^-(\theta)\in P^-(\theta)$ to be
the closest point to $\mu(\theta)$ among $P^-(\theta)$.
We then observe that the width of $D(\theta)$ is
\[w(\theta) = \max\{\sigma_\theta(q^+(\theta), \extr^+(\theta)),
    \sigma_\theta(q^-(\theta), \extr^-(\theta))\}.\]
Hence, $D(\theta)$ is completely determined by these four points:
$\extr^+(\theta)$, $\extr^-(\theta)$, $q^+(\theta)$, and $q^-(\theta)$.

In order to analyze and specify the change of pair $(q^+(\theta), q^-(\theta))$
as $\theta$ continuously increases from $-\pi/2$ to $\pi/2$,
we adopt an interpretation under a geometric dualization~\cite[Chapter 8]{bkos-cgaa-00}.
We shall call the plane $\Plane$---in which we have discussed objects so far---the \emph{primal plane} with the $x$- and $y$-axes.
Let $\DPlane$ be another plane, called the \emph{dual plane},
with $u$- and $v$-axes that correspond to its horizontal and vertical axes,
respectively.
Here, we use a standard duality transform $\dual{}$ which maps a point
$p = (a, b) \in \Plane$ into a line $\dual{p} \colon v = au - b \subset \DPlane$
and a non-vertical line $\ell \colon y = ax - b \subset \Plane$ into a point
$\dual{\ell} = (a, b) \in \DPlane$.
This duality transform is also defined in the reversed way
for points and lines in the dual plane $\DPlane$ to be mapped to
lines and points in the primal plane $\Plane$,
so that we have $\dual{(\dual{p})} = p$ and $\dual{(\dual{\ell})} = \ell$
for any point $p$ and any non-vertical line $\ell$
either in $\Plane$ or in $\DPlane$.
We say that a geometric object and its image under the duality transform
are \emph{dual} to each other.
Note that $p$ lies above (on or below, resp.) $\ell$ if and only if $\dual{\ell}$ lies above (on or below, resp.) $\dual{p}$.

\subsection{Scenes from the dual plane}

\begin{figure}[tb]
\begin{center}
\includegraphics[width=.99\textwidth]{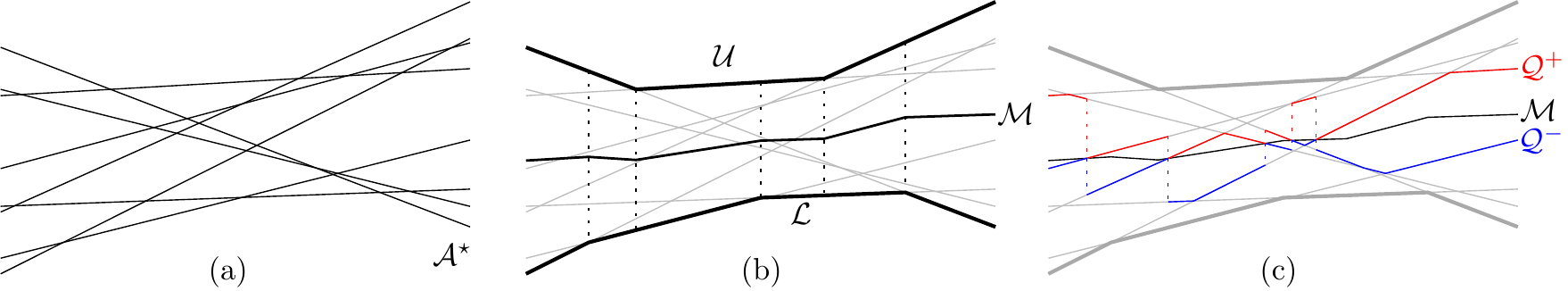}
\end{center}
\caption{(a) The arrangement $\darr$ of lines in $\dual{P}$.
 (b) The upper envelope $\uenv$, the lower envelope $\lenv$, and $\mline$
 of $\darr$.
 (c) $\qabove$ and $\qbelow$ depicted by red and blue chains, respectively.
 }
\label{fig:dual}
\end{figure}

Suppose that the input point set $P$ is given in the primal plane $\Plane$.
Consider the set $\dual{P} := \{\dual{p} \mid p \in P\}$ of $n$ lines
in the dual plane $\DPlane$
and their arrangement $\darr:=\arr(\dual{P})$.
See \figurename~\ref{fig:dual} for illustration.
Let $\uenv$ and $\lenv$ be the upper and lower envelopes in $\darr$.
The envelopes $\uenv$ and $\lenv$ can also be considered as two functions
of $u\in \Real$;
in this way, $\uenv(u)$ and $\lenv(u)$ are the $v$-coordinates in $\DPlane$
of points on $\uenv$ and $\lenv$, respectively, at $u\in \Real$.
Let $\mline(u):= (\uenv(u)+\lenv(u))/2$ be the $v$-coordinates of the midpoint
of the vertical segment connecting $\lenv$ and $\uenv$ at $u \in \Real$.
Similarly, we regard $\mline$ as the function itself and simultaneously
as its graph $\mline = \{(u,\mline(u)) \mid u\in \Real\}$ drawn in $\DPlane$.

As well known, the upper envelope $\uenv$ corresponds to
the lower chain of $\conv(P)$ and the lower envelope $\lenv$ to its upper chain.
More precisely, each vertex of $\uenv$ and $\lenv$ is dual to
the line supporting an edge of $\conv(P)$.
Thus, the total number of vertices of $\uenv$ and $\lenv$ is no more than $n=|P|$.
Also, observe that the number of vertices of $\mline$ is
equal to the total number of vertices of $\uenv$ and $\lenv$
by definition.

Now, consider the portions of lines in $\dual{P}$ above $\mline$
and the lower envelope of those pieces cut by $\mline$, denoted by $\qabove$.
Analogously,
let $\qbelow$ be the upper envelope of portions of lines in $\dual{P}$
below $\mline$.
The following observations follow directly from the basic properties of duality.
\begin{observation} \label{obs:dual}
 For each $\theta \in [-\pi/2, \pi/2)$, let $u := \tan \theta$.
 Then, the following hold:
 \begin{enumerate}[(1)] \denseitems
  \item The dual $\dual{(\ell^+(\theta))}$ of the $\theta$-aligned line
  through $\extr^+(\theta)$ is the point $(u, \lenv(u))$ in $\DPlane$.
  Similarly, we have $\dual{(\ell^-(\theta))} = (u, \uenv(u))$.
  \item The dual $\dual{(\mu(\theta))}$ of the middle line $\mu(\theta)$
  is the point $(u, \mline(u))$ in $\DPlane$.
  \item The dual of the $\theta$-aligned line through $q^+(\theta)$
  is the point $(u, \qbelow(u))$ in $\DPlane$.
  Similarly, the dual of the $\theta$-aligned line through $q^-(\theta)$
  is the point $(u, \qabove(u))$.
 \end{enumerate}
\end{observation}
From Observation~\ref{obs:dual},
one would say informally that
$\ell^+$ is dual to $\lenv$, $\ell^-$ to $\uenv$,
$\mu$ to $\mline$, $q^+$ to $\qbelow$, and $q^-$ to $\qabove$.

We are ready to describe our algorithm.
We first compute the arrangement $\darr$ in $O(n^2)$ time.
The envelopes $\uenv$ and $\lenv$ can be traced in $O(n)$ time from $\darr$,
and we also can compute $\mline$ in $O(n)$ time.
We then compute $\qabove$ and $\qbelow$.
\begin{lemma} \label{lem:Q}
 The complexity of $\qabove$ and $\qbelow$ is $O(n^2)$,
 and we can compute them in $O(n^2)$ time.
\end{lemma}
\begin{proof}
We make use of the Zone Theorem.
For any set $L$ of lines and a line $\ell$ in the plane,
the \emph{zone} of $\ell$ in the arrangement $\arr(L)$ of $L$ is
the set of cells of $\arr(L)$ that are intersected by $\ell$.
The Zone Theorem states that the number of vertices, edges, and cells
in the zone of $\ell$ in $\arr(L)$ is at most $5|L|$~\cite{bkos-cgaa-00}.

For each segment $e$ of $\mline$, we find all intersections
$e\cap \dual{p}$ for every $p \in P$.
There are at most $n$ such intersections for each $e$.
The part of $\qabove$ and $\qbelow$ above and below $e$, respectively,
can be simply by walking along the boundary of cells of $\darr$
intersected by $e$.
By the Zone Theorem, the number of vertices, edges, and cells of $\darr$
we traverse is bounded by $O(n)$ and the time spent for the walk
is also bounded by $O(n)$ for each segment $e$ of $\mline$,
see also the book~\cite[Chapter 8]{bkos-cgaa-00}.
Since $\mline$ consists of $O(n)$ segments, the total complexity to
explicitly construct $\qabove$ and $\qbelow$ is bounded by $O(n^2)$.
\end{proof}
By Lemma~\ref{lem:Q} together with Observation~\ref{obs:dual},
we know that there are $O(n^2)$ changes in pair $(q^+(\theta), q^-(\theta))$
as $\theta$ increases from $-\pi/2$ to $\pi/2$.
On the other hand, we already know that
the antipodal pair $(\extr^+(\theta), \extr^-(\theta))$ changes
$O(n)$ times as $\theta$ increases from $-\pi/2$ to $\pi/2$.
Consequently, there are $O(n^2)$ changes in tuple
$(\extr^+(\theta), \extr^-(\theta), q^+(\theta), q^-(\theta))$,
and thus the orientation domain $[-\pi/2, \pi/2)$ is decomposed into
$O(n^2)$ intervals in each of which the tuple is fixed.
For each such interval $I$,
we minimize $w(\theta) = \max\{\sigma_\theta(q^+(\theta), \extr^+(\theta)),
    \sigma_\theta(q^-(\theta), \extr^-(\theta))\}$ over $\theta \in I$.
Since the four points
$\extr^+(\theta), \extr^-(\theta), q^+(\theta), q^-(\theta)$ are fixed in $I$,
the function $w$ on $I$ is the upper envelope of at most
four sinusoidal functions by Lemma~\ref{lemma:sinusoidal}.
By Lemma~\ref{lem:double-strip-conf2},
the minimum occurs either (a) at an endpoint of $I$ or
(b) when the equality $\sigma_\theta(q^+(\theta), \extr^+(\theta)) =
    \sigma_\theta(q^-(\theta), \extr^-(\theta))$ holds.
Hence, we can minimize $w(\theta)$ over $\theta \in I$ in $O(1)$ time.
Since $w^* = \min_I \min_{\theta \in I} w(\theta)$,
we can compute $w^*$ by taking the minimum over such intervals $I$.

\begin{theorem} \label{thm:double_strip_AO}
 Given a set $P$ of $n$ points in the plane,
 a minimum-width double-strip enclosing $P$ can be computed in $O(n^2)$ time.
\end{theorem}

\section{Constrained Double-Strip Problem}
\label{sec:c_double-strip}
In this section, we discuss a constrained version of the minimum-width
double-strip problem, called the \emph{constrained double-strip problem}.
The constrained double-strip problem has its own interest,
while it can also be used, in particular,
to obtain efficient algorithms for the parallelogram annulus problem,
which will be discussed in the following section.

In the constrained double-strip problem,
we are given a set $P$ of $n$ points and a subset $Q \subseteq P$
with $k = |Q|$.
A \emph{$P$-constrained double-strip} is a double-strip
whose outer strip contains all points in $P$.
Then, the problem asks to find a $P$-constrained double-strip of minimum width
that encloses subset $Q$.

Analogously to Observation~\ref{obs:double-strip-conf},
we observe the following for the constrained problem.
\begin{observation} \label{obs:const_double-strip}
 For each $\theta \in [-\pi/2, \pi/2)$,
 there exists a minimum-width $\theta$-aligned $P$-constrained double-strip
 $D_{Q}(\theta)$ enclosing $Q$ such that its outer strip is $S(\theta)$ and
 its inner strip is $S'_{Q}(\theta)$,
 where $S'_{Q}(\theta)$ is the maximum-width $\theta$-aligned strip
  such that its interior is empty of any point in $Q$ and
  $\mu(S'_Q(\theta)) = \mu(\theta)$.
\end{observation}
\begin{proof}
The proof is almost identical to that of Observation~\ref{obs:double-strip-conf}.
Let $D$ be a minimum-width $\theta$-aligned $P$-constrained double-strip enclosing $P$.
Let $S_1$ and $S_2$ be the two strips such that $D = S_1 \cup S_2$.
Note that the outer strip of $D$ should enclose all points in $P$ by definition.
If the boundary of $S_1$ does not contain an extreme point of $P$,
then we can slide $S_1$ inwards until its boundary hits an extreme point of $P$.
Let $S'_1$ be the resulting strip after this sliding process.
Then, we have $S_1 \cap P \subset S'_1$.
In the same way, we slide $S_2$ until the boundary of $S_2$ hits
an extreme point of $P$, and let $S'_2$ be the resulting strip.
Then, it holds that $S_2 \cap P \subset S'_2$.
As a result, $P \subset (S'_1 \cup S'_2)$ since $P\subset (S_1 \cup S_2)$.
Note that the outer strip of the double-strip $D' := S'_1 \cup S'_2$
is exactly $S(\theta)$.

Now, let $S'$ be the inner strip of $D'$.
By definition, we have $\mu(S(\theta)) = \mu(S')$.
Suppose that the inner strip $S'$ of $D'$ is not equal to $S'_Q(\theta)$,
the maximum-width $\theta$-aligned strip
such that its interior is empty of any point in $Q$ and $\mu(S'(\theta)) = \mu(S(\theta))$.
Then, the boundary of $S'$ does not contain any point of $Q$.
Hence, the width of $S'_Q(\theta)$ is strictly larger than the width of $S'$,
a contradiction to the fact that $D'$ is of minimum width.
Therefore, the inner strip of $D'$ should be $S'_Q(\theta)$.
\end{proof}

Let $w_{Q}(\theta)$ be the width of the minimum-width $P$-constrained
double-strip $D_{Q}(\theta)$ enclosing $Q$
described in Observation~\ref{obs:const_double-strip}, and
$w^*_Q := \min_{\theta \in [-\pi/2, \pi/2)} w_Q(\theta)$.
Hence, we focus on computing the minimum possible width $w^*_Q$
and its corresponding double-strip $D_Q(w^*_Q)$.
We also redefine $q^+(\theta)$ and $q^-(\theta)$ to be
the closest points to $\mu(\theta)$ among points in $Q$
above and below $\mu(\theta)$, respectively.
Then, we have
 \[w_{Q}(\theta) = \max\{\sigma_\theta(q^+(\theta), \extr^+(\theta)),
    \sigma_\theta(q^-(\theta), \extr^-(\theta))\}.\]
The following lemma is analogous to Lemma~\ref{lem:double-strip-conf2}.
\begin{lemma} \label{lem:const_double-strip_conf}
 Let $\theta^*$ be an orientation such that $w_Q(\theta^*) = w^*_Q$.
 Then, either
 \begin{enumerate}[(a)] \denseitems
  \item three extreme points of $P$ lie on the boundary of $S(\theta^*)$, or
  \item two points of $Q$ lie on the boundary of $S'_Q(\theta^*)$.
 \end{enumerate}
\end{lemma}
\begin{proof}
The proof is almost identical to that of Lemma~\ref{lem:double-strip-conf2}.
From Observation~\ref{obs:const_double-strip},
we know that the boundary of $S(\theta^*)$ contains at least two extreme points,
$\extr^+(\theta^*)$ and $\extr^-(\theta^*)$,
while the boundary of $S'_Q(\theta^*)$ contains at least one point $q\in Q$
that maximizes $d_p(\theta^*)$ over all $p\in Q$.

Suppose that there is no more point in $P$ or $Q$ lying on the boundary
of $S(\theta^*)$ or of $S'(\theta^*)$, respectively.
Without loss of generality, we assume that
$q$ lies above $\mu(S(\theta))$.
This implies that the width of $D_Q(\theta^*)$ is
 \[ w^*_Q = w_Q(\theta^*) = d_q(\theta^*) = \sigma_{\theta^*}(q, \extr^+(\theta^*)).\]
Consider the function $d_q(\theta)$ of $\theta$ near $\theta^*$.
Note that $d_q(\theta^*) = \sigma_{\theta^*}(q, \extr^+(\theta^*)$ is
a sinusoidal function.

Since there is no more extreme point on the boundary of $S(\theta^*)$,
there is $\theta'$ sufficiently close to $\theta^*$
such that $d_q(\theta') < d_q(\theta^*)$.
In addition, since there is no more point on the boundary of $S'_Q(\theta^*)$,
we indeed have a strict inequality $d_p(\theta^*) < d_q(\theta^*)$
and hence $d_p(\theta') < d_q(\theta')$ for all $p\in Q \setminus \{q\}$.
This implies that
 \[ w_Q(\theta') = d_q(\theta') < d_q(\theta^*) = w_Q(\theta^*) = w^*_Q,\]
a contradiction.
Hence, there must be one more extreme point in $P$ on the boundary of $S(\theta^*)$
or one more point in $Q$ on the boundary of $S'_Q(\theta^*)$,
implying the lemma.
\end{proof}

To solve the constrained problem, we extend our approach for the unconstrained problem.
Define $\darr_Q := \arr(\dual{Q} \cup \uenv \cup \lenv)$
to be the arrangement of $k$ lines in $\dual{Q} = \{\dual{p} \mid p\in Q\}$ plus
the envelopes $\uenv$ and $\lenv$ of $\arr(\dual{P})$.
Let $\qabove_Q$ be the lower envelope of portions of lines in $\dual{Q}$
above $\mline$,
and $\qbelow_Q$ be the upper envelope of portions of lines in $\dual{Q}$
below $\mline$.

Our algorithm for the constrained double-strip problem runs as follows:
First, compute the convex hull $\conv(P)$ and extract
$\uenv$ and $\lenv$ from $\conv(P)$ in $O(n\log n)$ time.
Then, we add lines in $\dual{Q}$ incrementally one by one
to build $\darr_Q$.
Since every line in $\dual{Q}$ lies between $\uenv$ and $\lenv$,
this takes additional $O(k^2)$ time.
Next, we compute $\qabove_Q$ and $\qbelow_Q$.
\begin{lemma} \label{lem:Q_Q}
 The complexity of $\qabove_Q$ and $\qbelow_Q$ is $O(kn)$,
 and we can compute them in $O(kn)$ time.
\end{lemma}
\begin{proof}
The proof is almost identical to that of Lemma~\ref{lem:Q}.
The difference is the number of lines we consider is only $k \leq n$.

For each segment $e$ of $\mline$, we find all intersections
$e\cap \dual{p}$ for every $p \in Q$.
There are at most $k = |Q|$ such intersections for each $e$.
The part of $\qabove_Q$ and $\qbelow_Q$ above and below $e$, respectively,
can be simply by walking along the boundary of cells of $\darr_Q$
intersected by $e$.
By the Zone Theorem, the number of vertices, edges, and cells of $\darr_Q$
we traverse is bounded by $O(k)$ and the time spent for the walk
is also bounded by $O(k)$ for each segment $e$ of $\mline$,
see also the book~\cite[Chapter 8]{bkos-cgaa-00}.
Since $\mline$ consists of $O(n)$ segments, the total complexity to
explicitly construct $\qabove_Q$ and $\qbelow_Q$ is bounded by $O(kn)$.
\end{proof}
The rest of the algorithm is the same as that described in the previous section.
As we increase $\theta \in [-\pi/2, \pi/2)$ continuously,
there are $O(kn)$ changes in tuple
$(\extr^+(\theta), \extr^-(\theta), q^+(\theta), q^-(\theta))$
by Lemma~\ref{lem:Q_Q},
and thus the orientation domain $[-\pi/2, \pi/2)$ is decomposed into
$O(kn)$ intervals in each of which the four points are fixed.
For each such interval $I$,
we minimize $w_Q(\theta) = \max\{\sigma_\theta(q^+(\theta), \extr^+(\theta)),
    \sigma_\theta(q^-(\theta), \extr^-(\theta))\}$ over $\theta \in I$.
By Lemma~\ref{lem:const_double-strip_conf},
the minimum occurs either (a) at an endpoint of $I$ or
(b) when the equality $\sigma_\theta(q^+(\theta), \extr^+(\theta)) =
    \sigma_\theta(q^-(\theta), \extr^-(\theta))$ holds.
Hence, we can minimize $w(\theta)$ over $\theta \in I$ in $O(1)$ time,
and obtain the following result.
\begin{theorem} \label{thm:const-double_strip}
 Given a set $P$ of $n$ points and a subset $Q\subseteq P$
 with $k = |Q|$,
 a minimum-width $P$-constrained double-strip enclosing $Q$ can be
 computed in $O(n\log n + kn)$ time.
\end{theorem}

\subsection{Dynamic maintenance under insertion and deletion}
In the following, we consider a dynamic situation in which a point in $P$ can
be inserted into or deleted from $Q$.
Our goal is to report a minimum-width $P$-constrained double-strip enclosing $Q$
over all orientations and its width, that is, $w^*_Q$ and $D_Q(w^*_Q)$,
whenever a change in $Q$ occurs,
faster than computing it from scratch by Theorem~\ref{thm:const-double_strip}.

\begin{figure}[tb]
\begin{center}
\includegraphics[width=.90\textwidth]{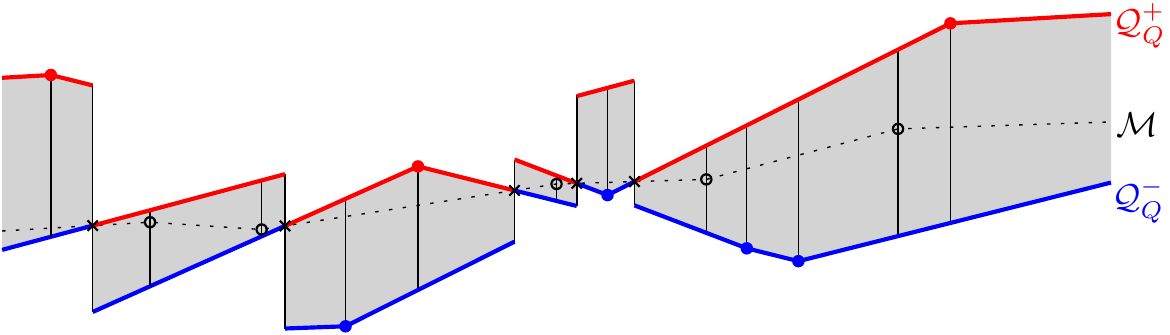}
\end{center}
\caption{Illustration to the trapezoidal map $\tmap_Q$ with $18$ trapezoids.
 In this example, we take $Q=P$ for input point set $P$, being the same
 as the one used as in \figurename~\ref{fig:dual}.
 Red and blue dots are vertices of $\darr_Q$ that belong to $\qabove_Q$ and $\qbelow_Q$, respectively.
 Small circles depict the vertices of $\mline$ and
 cross marks are the intersections between $\mline$ and any line $\dual{p}$
 for $p\in Q$.
 }
\label{fig:tmap}
\end{figure}

For the purpose, we keep the following invariants updated
under insertions and deletions: the arrangement $\darr_Q$ and
a trapezoidal map $\tmap_Q$.
The map $\tmap_Q$ is a vertical trapezoidal decomposition of the region between
$\qabove_Q$ and $\qbelow_Q$.
More precisely, let $\qbetween_Q$ be the region
between $\qabove_Q$ and $\qbelow_Q$,
and $\{u_1, u_2, \ldots, u_m\}$ be the set of $u$-coordinates of
the vertices of $\qabove_Q$, $\qbelow_Q$, and $\mline$
such that $u_1 < u_2 <\cdots < u_m$.
By adding into $\qbetween_Q$
a vertical line segment at $u=u_i$ for each $i=1,2, \ldots, m$
between two points $(u_i, \qbelow_Q(u_i))$ and $(u_i, \qabove_Q(u_i))$,
we obtain the trapezoidal decomposition $\tmap_Q$ consisting of $m+1$ trapezoids.
For convenience, let $u_0 := -\infty$ and $u_{m+1} := \infty$.
The order of trapezoids in $\tmap_Q$ is naturally induced
along the $u$-axis in $\DPlane$.
The $i$-th trapezoid $\tau$ in $\tmap_Q$ for $i=1, \ldots, m+1$
is bounded by two vertical segments at $u=u_{i-1}$ and $u=u_i$,
and two segments from $\qabove_Q$ and $\qbelow_Q$.
See \figurename~\ref{fig:tmap} for an illustration.
The two segments of $\tau$ from $\qabove_Q$ and $\qbelow_Q$
are called the \emph{ceiling} and \emph{floor} of $\tau$, respectively.
Let $U_\tau \subset \Real$ be an interval consisting of the $u$-coordinates
of all points in $\tau$.

At each trapezoid $\tau$, we store
four points $\extr^+_\tau, \extr^-_\tau, q^+_\tau, q^-_\tau \in P$
such that $\extr^+_\tau=\extr^+(\theta)$, $\extr^-_\tau = \extr^-(\theta)$,
$q^+_\tau = q^+(\theta)$, and $q^-_\tau = q^-(\theta)$ for all $\theta$
with $\tan \theta \in U_\tau$.
We also store at $\tau$ the value
  \[ w_\tau := \min_{\theta \colon \tan \theta \in U_\tau}
   \max\{\sigma_\theta(q^+_\tau, \extr^+_\tau), \sigma_\theta(q^-_\tau, \extr^-_\tau)\},\]
which can be computed in $O(1)$ time per trapezoid by Lemma~\ref{lem:const_double-strip_conf}.

Note that the union of ceilings of all trapezoids in $\tmap_Q$
forms $\qabove_Q$, and the union of their floors forms $\qbelow_Q$.
By Lemma~\ref{lem:Q_Q}, the number $m+1$ of trapezoids in $\tmap_Q$ is $O(|Q|n)$.
More importantly, for each $\theta \in [-\pi/2, \pi/2)$,
there is a unique trapezoid $\tau$ in $\tmap_Q$ such that
$\tan \theta$ lies in interval $U_\tau$ of $\tau$ and thus we have
$w_Q(\theta) = \max\{\sigma_\theta(q^+_\tau, \extr^+_\tau), \sigma_\theta(q^-_\tau, \extr^-_\tau)\}$.
This implies that $\min \{ w_Q(\theta) \mid \tan\theta \in U_\tau \} = w_\tau$,
and hence $w^*_Q = \min_{\tau \in \tmap_Q} w_\tau$.
Thus, by efficiently maintaining $\tmap_Q$,
the problem is reduced to finding the minimum of $w_\tau$ over
all trapezoids $\tau$ in $\tmap_Q$.

Updating our invariants can be done in $O(n)$ time
thanks to the Zone Theorem for the arrangement of lines.
\begin{lemma} \label{lem:invariants}
 Let $Q \subseteq P$, $p\in P$, and $Q'$ be
 either $Q \cup \{p\}$ or $Q \setminus \{p\}$.
 Then,
 $\darr_{Q'}$ and $\tmap_{Q'}$ can be obtained in $O(n)$ time,
 provided the description of $\darr_Q$ and $\tmap_Q$.
 More specifically, the number of trapezoids in the symmetric difference
 $(\tmap_{Q'} \setminus \tmap_Q) \cup (\tmap_Q \setminus \tmap_{Q'})$
 is $O(n)$, and those trapezoids can be identified in the same time bound.
\end{lemma}
\begin{proof}
Here, we describe a procedure that updates our invariants
$\darr_Q$ and $\tmap_Q$ into $\darr_{Q'}$ and $\tmap_{Q'}$.

We first consider the case of insertion.
Let $p\in P \setminus Q$ and $Q' := Q \cup \{p\}$.
Adding $\ell = \dual{p}$ into $\darr_Q$ results in $\darr_{Q'}$.
The number of new vertices in $\darr_{Q'}$ is exactly $|Q| \leq n$,
and can be obtained by walking along $\ell$ through $\darr_Q$.
In this process, we traverse all the vertices, edges, and cells of $\darr_Q$
that form the zone of $\ell$ in $\darr_Q$.
Hence, by the Zone Theorem, we obtain $\darr_{Q'}$ from $\darr_Q$ in $O(n)$ time.

Also, in this process, we mark those segments of $\qabove_Q$ and $\qbelow_Q$
that are crossed by $\ell$.
From these marked segments, we can identify
all trapezoids in $\tmap_Q$ crossed by $\ell$,
and those trapezoids form the set $\tmap_Q \setminus \tmap_{Q'}$.
Observe that each trapezoid $\tau \in \tmap_Q \setminus \tmap_{Q'}$
is crossed by $\ell$,
so $\tau$ is a subset of a cell in the zone of $\ell$ in $\darr_Q$.
Hence, we have $|\tmap_Q \setminus \tmap_{Q'}| = O(|Q|) = O(n)$.
Every new trapezoid $\tau$ in $\tmap_{Q'} \setminus \tmap_Q$
can be computed by walking along $\ell$
such that $U_\tau$ is determined by new vertices in $\qabove_Q$ and $\qbelow_Q$,
including the intersection points between $\ell$ and $\mline$.
So, we have $|\tmap_{Q'} \setminus \tmap_Q| = O(n)$ and
we can specify all new trapezoids in $\tmap_{Q'} \setminus \tmap_Q$
in the same time bound.
To obtain $\tmap_{Q'}$ from $\tmap_Q$,
we can simply remove those in $\tmap_Q \setminus \tmap_{Q'}$
and add those in $\tmap_{Q'} \setminus \tmap_Q$.

Next, consider the case of deletion.
Let $p \in Q$ and $Q' := Q\setminus \{p\}$.
Also, let $\ell := \dual{p}$.
Again, by the Zone Theorem,
the total number of vertices and edges in the cells of $\darr_Q$
adjacent to $\ell$ is bounded by $O(n)$.
To obtain $\darr_{Q'}$ from $\darr_Q$,
we remove all edges $e$ of $\darr_Q$ that are portions of $\ell$
and merge two cells adjacent to $e$ into one cell.
This can be done in $O(n)$ time by the Zone Theorem.

In order to update the trapezoidal map,
observe that for each trapezoid in $\tmap_Q \setminus \tmap_{Q'}$
its ceiling or floor is a segment of $\ell$.
Hence, the number of those is $O(n)$ and can be specified in $O(n)$ time
from the arrangement $\darr_Q$.
On the other hand, new trapezoids in $\tmap_{Q'} \setminus \tmap_Q$
lie in a cell in the zone of $\ell$ in $\darr_{Q'}$.
Hence, their number is also bounded by $O(n)$ by the Zone Theorem,
and can be specified in the same time bound.
\end{proof}

Now, we are ready to describe our overall algorithm.
We assume that we start with an empty set $Q=\emptyset$
and a sequence of insertions and deletions on $Q$ is given.
For $Q=\emptyset$, it is easy to initialize $\darr_\emptyset$
and $\tmap_\emptyset$,
after computing $\uenv$, $\lenv$, and $\mline$ in $O(n\log n)$ time
as described above.
Namely, $\darr_\emptyset$ is just the union of $\uenv$ and $\lenv$,
and $\tmap_\emptyset$ can be obtained from the fact that
$\qabove_\emptyset = \uenv$ and $\qbelow_\emptyset = \lenv$.
Whenever an insertion or deletion of a point is given,
we update our invariants as described in Lemma~\ref{lem:invariants},
spending $O(n)$ time.
We then report the minimum possible width $w^*_Q$ for the current subset $Q$
and the corresponding double-strip.
Since $\tmap_Q$ consists of $O(|Q|n)$ trapezoids by Lemma~\ref{lem:Q_Q},
a linear scan of $\tmap_Q$ is too costly.
We instead use a basic priority queue, such as the binary heap, and conclude the following.
\begin{theorem} \label{thm:const-double_strip-online}
 Let $P$ be a set of $n$ points, and $Q_0 = \emptyset, Q_1, Q_2, \ldots$ be
 a sequence of subsets of $P$ such that the difference between
 $Q_{i+1}$ and $Q_i$ is only a single point in $P$.
 Suppose that each $Q_i$ is given at time $i$ by its difference from $Q_{i-1}$.
 Then, whenever $Q_i$ is specified for each $i \geq 0$, we can exactly compute
 a $P$-constrained double-strip of minimum width $w^*_{Q_i}$ that encloses
 $Q_i$ in $O(n \log n)$ time.
\end{theorem}
\begin{proof}
We additionally maintain a priority queue $W$, along with $\darr_Q$ and $\tmap_Q$.
The priority queue $W$ consists of the trapezoids $\tau_i$ of $\tmap_Q$,
indexed by the values of $w_i$, and supports the following operations
in $O(\log |W|)$ time:
inserting a trapezoid into $W$,
deleting a trapezoid from $W$, and
finding one with smallest $w_i$-value among those in $W$.
This can be achieved by a binary heap structure with $O(|W|)$ additional pointers
between the trapezoids of $\tmap_Q$ and those of $W$.

For $Q=Q_0 = \emptyset$,
$W$ can be set up by inserting $O(n)$ trapezoids of $\tmap_Q$ in $O(n \log n)$ time.

Suppose that we have correctly maintained
$\darr_{Q_i}$, $\tmap_{Q_i}$, and $W$, and we are now given $Q_{i+1}$.
We then compute $\darr_{Q_{i+1}}$ and $\tmap_{Q_{i+1}}$ from
$\darr_{Q_i}$ ad $\tmap_{Q_i}$ in $O(n)$ time described in Lemma~\ref{lem:invariants}.
At the same time, we specify those trapezoids in the symmetric difference
between $\tmap_{Q_i}$ and $\tmap_{Q_{i+1}}$.
So, we delete trapezoids $\tau \in \tmap_{Q_{i}} \setminus \tmap_{Q_{i+1}}$
from $W$, and insert those $\tau \in \tmap_{Q_{i+1}} \setminus \tmap_{Q_i}$
into $W$.
Again by Lemma~\ref{lem:invariants}, the number of trapezoids that are
either deleted from or inserted into $W$ is bounded by $O(n)$.
Hence, we can update $W$ for $Q_{i+1}$ in $O(n\log n)$ time.

To report the minimum-width $P$-constrained double-strip enclosing $Q_{i+1}$,
we find the minimum element in $W$.
Thus, the theorem follows.
\end{proof}

If one only wants to decide whether or not the minimum possible width $w^*_Q$
is at least a given target value $w \geq 0$,
then the complexity can be reduced as follows.
\begin{theorem} \label{thm:const-double_strip-online-decision}
 Let $P$ be a set of $n$ points, and $Q_0 = \emptyset, Q_1, Q_2, \ldots$ be
 a sequence of subsets of $P$ such that the difference between
 $Q_{i+1}$ and $Q_i$ is only a single point in $P$.
 Suppose that each $Q_i$ is given at time $i$ by its difference from $Q_{i-1}$.
 Let $w \geq 0$ be a given fixed real number.
 Then, after spending $O(n \log n)$ time for preprocessing,
 whenever $Q_i$ is specified for each $i \geq 0$,
 we can decide whether $w \geq w^*_{Q_i}$ or $w < w^*_{Q_i}$  in $O(n)$ time.
\end{theorem}
\begin{proof}
The proof is almost identical to that of Theorem~\ref{thm:const-double_strip-online} except that we do not need to maintain the priority queue $W$ for decision.
Instead, we maintain a counter variable $c$.

The counter $c$ stores the number of trapezoids $\tau$ in $\tmap_Q$
such that $w_\tau \leq w$.
It is not difficult to maintain the counter $c$ in $O(n)$ time
since the number of trapezoids in the symmetric difference between
$\tmap_{Q_{i-1}}$ and $\tmap_{Q_i}$ is bounded by $O(n)$
by Lemma~\ref{lem:invariants}.
Then, it holds that
$w \geq w^*_{Q_i}$ if and only if $c > 0$;
while $w < w^*_{Q_i}$ if and only if $c = 0$.
\end{proof}

\subsection{Offline version under insertions only}
Note that the above theorems give us solutions to
the online optimization and decision versions of
the $P$-constrained double-strip problem under insertions and deletions.
Here, we consider the offline optimization version of the problem
under insertions only.

Let $Q = \{p_1, p_2, \ldots, p_k\} \subseteq P$ be a subset of $P$,
and $Q_i := \{p_1, \ldots, p_i\}$ for $i=0, \ldots, k$.
Suppose that we know $Q_i$ for each $i=0, \ldots, k$ for the first time and
want to compute a minimum-width $P$-constrained double-strip
enclosing $Q_i$ for all $i=0, \ldots, k$.

For the purpose, we observe the following.
\begin{lemma} \label{lem:backward}
 For each $i=0, \ldots, k-1$, it holds that
 \[ w^*_{Q_i} = \min \{ w^*_{Q_{i+1}}, \min_{\tau\in T_i} w_\tau \},  \]
 where $T_i := \tmap_{Q_i} \setminus \tmap_{Q_{i+1}}$ denotes
 the set of trapezoids removed
 from $\tmap_{Q_i}$ by the insertion of $p_{i+1}$.
\end{lemma}
\begin{proof}
First, we show that we always have $w^*_{Q_{i}} \leq w^*_{Q_{i+1}}$.
Consider a $P$-constrained double-strip $D$ enclosing $Q_{i+1}$
of minimum possible width $w^*_{Q_{i+1}}$.
Since $Q_i \subset Q_{i+1} \subset D$,
$D$ is also a $P$-constrained double-strip enclosing $Q_i$.
This implies that $w^*_{Q_{i}} \leq w^*_{Q_{i+1}}$.

Let $T'_i := \tmap_{Q_{i+1}} \setminus \tmap_{Q_{i}}$,
and $S_i := \tmap_{Q_i} \cap \tmap_{Q_{i+1}}$.
Then, $\tmap_{Q_i} = S_i \cup T_i$ and $\tmap_{Q_{i+1}} = S_i \cup T'_i$.
Letting $s := \min_{\tau\in S_i} w_\tau$
$t := \min_{\tau \in T_i} w_\tau$, and $t' := \min_{\tau \in T'_i} w_\tau$,
we have
 \[ w^*_{Q_i} = \min_{\tau \in \tmap_{Q_i}} w_\tau = \min\{s, t\}
 \quad \text{ and } \quad
  w^*_{Q_{i+1}} = \min_{\tau \in \tmap_{Q_{i+1}}} w_\tau = \min\{s, t'\}.\]
Note that we also have
$w^*_{Q_i} \leq w^*_{Q_{i+1}}$, and thus
$\min\{s, t\} \leq \min\{s, t'\}$.
The lemma claims that $w^*_{Q_i} = \min \{s, t, t'\}$.

We consider all possible cases of the relative order
among the three values $s$, $t$, and $t'$, and prove the lemma for every case.
If $s$ is the smallest among the three values $\{s, t, t'\}$,
then we have  $w^*_{Q_i} = s = \min\{s, t, t'\}$, so the lemma follows.
If $t$ is the smallest among the three values $\{s, t, t' \}$,
then we have  $w^*_{Q_i} = t = \min\{s, t, t'\}$, so the lemma follows.

Lastly, suppose that neither $s$ nor $t$ is the smallest among the three.
This means that $t' < s$ and $t' < t$.
We then have
\[ w^*_{Q_{i+1}} = t' < \min\{s, t\} = w^*_{Q_i}, \]
a contradiction to our observation $w^*_{Q_{i}} \leq w^*_{Q_{i+1}}$
shown above.
Therefore, either $s$ or $t$ is always the smallest among $\{s, t, t'\}$,
and hence we always have $w^*_{Q_i} = \min \{s, t, t'\}$,
as claimed.
\end{proof}

Lemma~\ref{lem:backward} suggests computing $w^*_{Q_i}$
backwards from $i = k$ to $i=0$.
By maintaining $\tmap_{Q_i}$ from $i=0$ to $k$ and
storing the sets $T_i = \tmap_{Q_i} \setminus \tmap_{Q_{i+1}}$,
this can be done in $O(kn)$ time.

More precisely, we first build $\darr_{Q_0}$ and $\tmap_{Q_0}$ as described above
in $O(n\log n)$ time.
We then insert $p_i$ for $i=1, \ldots, k$, one by one,
and compute $\darr_{Q_i}$ and $\tmap_{Q_i}$ in $O(n)$ time per insertion
by Lemma~\ref{lem:invariants},
but we do not compute the minimum width $w^*_{Q_i}$ at every insertion.
Instead, we collect all trapezoids $\tau$ that have been deleted,
that is, the set $T_i = \tmap_{Q_i} \setminus \tmap_{Q_{i+1}}$.
Then, we apply Lemma~\ref{lem:backward} to compute
$w^*_{Q_i}$ for each $i=0, \ldots, k$ and its corresponding double-strip.

We first compute $w^*_{Q_k} = \min_{\tau \in \tmap_{Q_k}} w_\tau$.
Then, we iterate $i$ from $k-1$ to $0$, and compute
$w^*_{Q_i}$ based on Lemma~\ref{lem:backward}.
Since $|T_i| = O(n)$ for each $i$ by Lemma~\ref{lem:invariants},
this takes $O(kn)$ additional time.
We thus conclude the following theorem.

\begin{theorem} \label{thm:const_double-strip_offline_insertion}
 Let $P$ be a set of $n$ points and $p_1, \ldots, p_k \in P$
 be $k \geq 1$ points in $P$, and
 let $Q_i := \{p_1, \ldots, p_i\}$ for $0 \leq i \leq k$.
 Then, in time $O(n \log n + kn)$, we can exactly compute
 $w^*_{Q_i}$ for all $0 \leq i \leq k$
 and corresponding $P$-constrained double-strips of width $w^*_{Q_i}$
 enclosing $Q_i$.
\end{theorem}
\begin{proof}
We first build $\darr_{Q_0}$ and $\tmap_{Q_0}$ as described above
in $O(n\log n)$ time.
We then insert $p_i$ for $i=1, \ldots, k$, one by one,
and compute $\darr_{Q_i}$ and $\tmap_{Q_i}$ in $O(n)$ time per insertion
by Lemma~\ref{lem:invariants},
but we do not compute the minimum width $w^*_{Q_i}$ at every insertion.
Instead, we collect all trapezoids $\tau$ that has been deleted,
that is, the set $T_i = \tmap_{Q_i} \setminus \tmap_{Q_{i+1}}$.
Then, we apply Lemma~\ref{lem:backward} to compute
$w^*_{Q_i}$ for each $i=0, \ldots, k$ and its corresponding double-strip.

We first compute $w^*_{Q_k} = \min_{\tau \in \tmap_{Q_k}} w_\tau$.
Then, we iterate $i$ from $k-1$ to $0$, and compute
$w^*_{Q_i}$ based on Lemma~\ref{lem:backward}.
Since $|T_i| = O(n)$ for each $i$ by Lemma~\ref{lem:invariants},
this takes $O(kn)$ additional time.
We thus conclude the theorem.
\end{proof}

\section{Minimum-Width Parallelogram Annuli} \label{sec:paralannul}
In this section, we present algorithms that compute
a minimum-width parallelogram annulus that encloses
a set $P$ of $n$ points in $\Plane$.
As introduced in Section~\ref{sec:pre},
a parallelogram annulus is defined by two double-strips and
its orientation is represented by a pair of parameters $(\theta, \phi)$
with $\theta, \phi \in [-\pi/2, \pi/2)$.

Here, we consider several cases depending on how many
of the two side orientations, $\theta$ and $\phi$, are fixed or not.
The easiest case is certainly when both $\theta$ and $\phi$ are fixed.

\begin{observation} \label{obs:p_annulus}
 For any $\theta, \phi \in [-\pi/2, \pi/2)$,
 there exists a minimum-width $(\theta, \phi)$-aligned
 parallelogram annulus that encloses $P$
 such that its outer parallelogram $R(\theta, \phi)$ is
 the intersection of $S(\theta)$ and $S(\phi)$,
 the minimum-width $\theta$-aligned and $\phi$-aligned strip enclosing $P$,
 respectively.
\end{observation}
\begin{proof}
Let $A$ be any minimum-width $(\theta, \phi)$-aligned
parallelogram annulus that encloses $P$.
By definition, $A$ is defined by two double-strips $D_1$ and $D_2$,
where $D_1$ is $\theta$-aligned and $D_2$ is $\phi$-aligned.
Let $S_1$ and $S_2$ be the outer strips of $D_1$ and $D_2$, respectively.
Note that the outer parallelogram $R$ of $A$ is the intersection of
the outer strips $S_1$ and $S_2$ of $D_1$ and $D_2$,
that is, $R = S_1 \cap S_2$.
Since $A$ encloses $P$ and thus $R$ encloses $P$,
we have $P \subset S_1$ and $P\subset S_2$.
This implies that
$D_1$ is a $P$-constrained double-strip enclosing subset $P_1 := D_1\cap P$
and $D_2$ is a $P$-constrained double-strip enclosing subset $P_2 := D_2 \cap P$.

By Observation~\ref{obs:const_double-strip},
there exists the minimum-width $P$-constrained $\theta$-aligned
double-strip $D_{P_1}(\theta)$ enclosing $P_1$ whose outer strip is
equal to $S(\theta)$.
Analogously, we consider the minimum-width $P$-constrained $\phi$-aligned
double-strip $D_{P_2}(\phi)$ enclosing $P_2$ whose outer strip is
equal to $S(\phi)$.

Now consider another parallelogram annulus $A'$ defined by
$D_{P_1}(\theta)$ and $D_{P_2}(\theta)$.
Since $P_1 \subset D_{P_1}(\theta)$ and $P_2 \subset D_{P_2}(\phi)$,
it holds that $A'$ encloses $P$.
By definition of $D_{P_1}(\theta)$ and $D_{P_2}(\phi)$,
we have that the width of $A'$ is at most the width of $A$.
Hence, $A'$ is also a minimum-width $(\theta, \phi)$-aligned
parallelogram annulus that encloses $P$.
Finally, observe that the outer parallelogram of $A'$ is
exactly $R(\theta, \phi) = S(\theta) \cap S(\phi)$.
This proves the observation.
\end{proof}

The above observation gives us a structural property of an optimal annulus
which we should look for, and leads to a linear-time algorithm.

\begin{theorem} \label{thm:p_annulus_ff}
 Given a set $P$ of $n$ points and $\theta, \phi \in [-\pi/2, \pi/2)$,
 a minimum-with $(\theta, \phi)$-aligned parallelogram annulus that encloses $P$
 can be computed in $O(n)$ time.
\end{theorem}
\begin{proof}
Here, we describe an algorithm that finds
a minimum-width $(\theta, \phi)$-aligned parallelogram annulus enclosing $P$
whose outer parallelogram is $R(\theta, \phi) = S(\theta) \cap S(\phi)$,
as described in Observation~\ref{obs:p_annulus}.

As described in the proof of Theorem~\ref{thm:double-strip-fixed-orientation},
we can find the minimum-width $\theta$-aligned strip $S(\theta)$
enclosing $P$ in $O(n)$ time
by identifying the corresponding
antipodal pair $(\extr^+(\theta), \extr^-(\theta))$.
In the same way, we identify the antipodal pair $(\extr^+(\phi), \extr^-(\phi))$
and strip $S(\phi)$.

We then need to find an inner parallelogram $R'$ that minimizes the width
of the resulting annulus defined by $R(\theta,\phi)$ and $R'$.
For the purpose, we are done by checking the distance from each $p\in P$
to the boundary of $R$,
which is equal to
 \[ \min\{w_\theta(p, \extr^+(\theta)), w_\theta(p, \extr^-(\theta)),
     w_\phi(p, \extr^+(\phi)), w_\phi(p, \extr^-(\phi)) \}. \]
We can evaluate this in $O(1)$ time for each $p \in P$,
so in $O(n)$ total time, and take the maximum over them, denoted by $z$.
Since the interior of the inner parallelogram $R'$ must avoid all points in $P$,
at least one side of $R'$ must be $z$ distant from the boundary of $R$.
Hence, the minimum width of a $(\theta, \phi)$-aligned parallelogram annulus
enclosing $P$ is exactly $z$.
The value of $z$ and a corresponding annulus can be computed in $O(n)$ time.
\end{proof}

In the following,
let $w(\theta, \phi)$ be the smallest among the widths of
all $(\theta,\phi)$-aligned parallelogram annuli enclosing $P$.

\subsection{When one side orientation is fixed}

Next, we consider the problem where one side of a resulting annulus
should be $\phi$-aligned for a fixed orientation $\phi \in [-\pi/2, \pi/2)$
while the other orientation parameter $\theta$ can be chosen arbitrarily.
So, in the following, we regard $\phi \in [-\pi/2, \pi/2)$ to be fixed.
Without loss of generality, we assume that $\phi = 0$.

From the definition of a parallelogram annulus $A$,
it is defined by two double-strips.
In addition, Observation~\ref{obs:p_annulus} tells us that
the two double-strips defining $A$ can be chosen among
the $P$-constrained double-trips enclosing a subset of $P$.
Hence, for the case where $\phi = 0$ is fixed,
the problem is reduced to find a best bipartition of $P$ such that
one part is covered by a $0$-aligned $P$-constrained double-strip
and the other by another $P$-constrained double-strip
in any orientation $\theta \in [-\pi/2, \pi/2)$.

We first identify two extreme points $\extr^+(0)$ and $\extr^-(0)$,
and the strip $S(0)$ in $O(n)$ time.
Then, sort the points in $P$ in the non-increasing order of the value
 \[ d_p(0) = \min\{\sigma_0(p, \extr^+(0)), \sigma_0(p, \extr^-(0))\} \]
for each $p\in P$, which is the distance to the boundary of $S(0)$.
Let $p_1, p_2, \ldots, p_{n-1}, p_n \in P$ be this order.
Also, let $w_i := d_{p_i}(0) = \min\{\sigma_0(p_i, \extr^+(0)), \sigma_0(p_i, \extr^-(0))\}$
for $i=1, \ldots, n$.

Consider the double-strip $D_i$ with width $w_i$ and outer strip $S(0)$.
The double-strip $D_1$ encloses all points of $P$,
while $D_2$ misses one point $p_1$ and $D_i$ misses $i-1$ points $p_1, \ldots, p_{i-1}$ in general for $i=1,\ldots, n$.
This means that there are only $n$ different subsets of $P$
covered by any $0$-aligned double-strip.
Thus, to enclose $P$ by a $(\theta, 0)$-aligned parallelogram annulus,
the other double-strip with orientation $\theta$ should cover the rest
of the points in $P$.

Note that each $D_i$ is a minimum-width $0$-aligned $P$-constrained double-strip
that encloses $\{p_i, p_{i+1}, \ldots, p_n\} \subseteq P$.
Let $Q_i := \{p_1, p_2, \ldots, p_{i-1}\}$ for $i=1, \ldots, n$.
If we choose $D_i$ for the $0$-aligned double-strip,
then $P \setminus Q_i \subset D_i$,
so the points in $Q_i$ should be covered by the second double-strip
that define a parallelogram annulus.
Let $D'_i$ be the minimum-width $P$-constrained double-strip enclosing $Q_i$,
and let $w'_i$ be its width.
We compute $D'_i$ and $w'_i$ for all $i\in\{1,\ldots, n\}$ by applying Theorem~\ref{thm:const_double-strip_offline_insertion}
in $O(n^2)$ time.
What remains is taking the minimum of $\max\{w_i, w'_i\}$ over $i=1, \ldots, n$.

\begin{theorem} \label{thm:p_annulus_fa}
 Given a set $P$ of $n$ points and a fixed orientation $\phi \in [-\pi/2, \pi/2)$,
 a $(\theta, \phi)$-aligned parallelogram annulus of minimum width
 over all $\theta \in [-\pi/2,\pi/2)$
 that encloses $P$ can be computed in $O(n^2)$ time.
\end{theorem}

\subsection{General case}
Finally, we consider the general case where both $\theta$ and $\phi$
can be freely chosen from domain $[-\pi/2, \pi/2)$.
Let $w^*:=\min_{\theta, \phi \in [-\pi/2, \pi/2)} w(\theta,\phi)$
be the minimum possible width,
and $(\theta^*, \phi^*)$ be a pair of orientations such that
$w^* = w(\theta^*, \phi^*)$.

We first consider the decision version of the problem
in which a positive real number $w>0$ is given
and we want to decide if $w \geq w^*$ or $w < w^*$.
For the purpose, we consider the function $d_p$ defined above for each $p\in P$
to be
$d_p(\theta) = \min\{\sigma_\theta(p, \extr^+(\theta)), \sigma_\theta(p, \extr^-(\theta))\}$.
As observed above, the function $d_p$ is piecewise sinusoidal
with $O(n)$ breakpoints, so its graph
$\{ (\theta, y) \mid y = d_p(\theta), -\pi/2 \leq \theta < \pi/2\}$
consists of $O(n)$ sinusoidal curves.
Let $\Gamma_p$ be the set of these sinusoidal curves,
and $\Gamma := \bigcup_{p\in P} \Gamma_p$.
We build the arrangement $\arr(\Gamma)$ of these sinusoidal curves in $\Gamma$.
Note that each vertex of $\arr(\Gamma)$ corresponds either to
a breakpoint of function $d_p$ for some $p\in P$ or to
an intersection point between a curve in $\Gamma_p$ and another in $\Gamma_{p'}$
for some $p, p'\in P$ with $p\neq p'$.
\begin{lemma} \label{lem:arrG}
 The arrangement $\arr(\Gamma)$ of curves in $\Gamma$ consists of
 $O(n^3)$ vertices, edges, and cells,
 and can be computed in $O(n^3)$ time.
\end{lemma}
\begin{proof}
Consider the decomposition of the orientation domain $[-\pi/2, \pi/2)$
into maximal intervals $I$ such that
the antipodal pair $(\extr^+(\theta), \extr^-(\theta))$ is fixed
for any $\theta \in I$.
By Toussaint~\cite{t-sgprc-83}, there are $O(n)$ such intervals $I$.

We consider the arrangement $\arr(\Gamma)$ on each such interval $I$.
For each such interval $I$,
since the pair $(\extr^+(\theta), \extr^-(\theta))$ is fixed,
the function
  \[ d_p(\theta) = \min \{\sigma_\theta(p, \extr^+(\theta)),
     \sigma_\theta(p, \extr^-(\theta))\}\]
is piecewise sinusoidal with at most one breakpoint by Lemma~\ref{lemma:sinusoidal}.
Again by Lemma~\ref{lemma:sinusoidal}, any two curves in $\Gamma$ defined on $I$
cross at most once.
Hence, the complexity of the arrangement $\arr(G)$ on $I$ is $O(n^2)$,
and can be computed in $O(n^2)$ time.
We obtain the $O(n^3)$ bound by simply iterating all the $O(n)$ intervals $I$.
\end{proof}

Now, we describe our decision algorithm.
Let $w > 0$ be a given positive real number.
First, we intersect the horizontal line $\ell \colon \{y=w\}$ with
arrangement $\arr(\Gamma)$.
\begin{lemma} \label{lem:arrG-zone}
 Any horizontal line crosses the edges of $\arr(\Gamma)$ in $O(n^2)$ points,
 and all these intersection points can be specified in $O(n^2)$ time.
\end{lemma}
\begin{proof}
Consider the intervals $I$ described in the proof of Lemma~\ref{lem:arrG}.
For each $I$, the antipodal pair $(\extr^+(\theta), \extr^-(\theta))$
is fixed for all $\theta \in I$ and so
the function
  \[ d_p(\theta) = \min \{\sigma_\theta(p, \extr^+(\theta)),
     \sigma_\theta(p, \extr^-(\theta))\}\]
is piecewise sinusoidal with at most one breakpoint by Lemma~\ref{lemma:sinusoidal} for each $p\in P$.
Hence, any horizontal line $\ell$ crosses $d_p(\theta)$ at most $O(1)$ time
for $\theta \in I$.
Since there are $O(n)$ such intervals $I$,
the function $d_p$ for each $p\in P$ crosses $\ell$ in $O(n)$ intersection points.
Also, these intersection points can be computed in $O(n)$ time.

Since the edges of $\arr(\Gamma)$ come from the graph of $d_p$ for all $p\in P$,
we conclude that any horizontal line $\ell$ crosses the edges of $\arr(\Gamma)$
in $O(n^2)$ points, and all the intersection points can be specified
in the same time bound.
\end{proof}

Our algorithm continuously increases $\theta$ from $-\pi/2$ to $\pi/2$
and checks if there exists a parallelogram annulus of width $w$ that encloses $P$
such that one of the two double-strips defining it is $\theta$-aligned.

Let $\{\theta_1, \ldots, \theta_m\}$ be
the set of $\theta$-values of every intersection point
between $\ell$ and the edges of $\arr(\Gamma)$
such that $-\pi/2 \leq \theta_1 < \theta_2 < \cdots < \theta_m < \pi/2$.
Note that $m = O(n^2)$ by Lemma~\ref{lem:arrG-zone}.
For each $i=1, \ldots, m$,
let $P_i \subseteq P$ be the set of points $p\in P$
such that $d_p(\theta_i) \leq w$, and let $Q_i := P \setminus P_i$.
Let $D_i$ be the $\theta_i$-aligned $P$-constrained double-strip of width $w$.
Then, we have $P_i \subset D_i$ while $Q_i \cap D_i = \emptyset$.
Let $D'_i$ be a $P$-constrained double-strip of minimum width
that encloses $Q_i$.
Recall that the width of $D'_i$ is denoted by $w^*_{Q_i}$ in the previous section.
If the width $w^*_{Q_i}$ of $D'_i$ is at most $w$,
then the parallelogram annulus defined by $D_i$ and $D'_i$ indeed encloses $P$
and its width is $w$,
so we conclude that $w \geq w^*$.
Otherwise, if $w^*_{Q_i} > w$, then we proceed to next $\theta$-value
$\theta_{i+1}$.

We perform this test for each $i=1, \ldots, m$ in an efficient way
with the aid of our online decision algorithm for the constrained double-strip
problem.
Initially, for $i=1$, we compute $P_1$, $Q_1$, and $D_1$ in $O(n)$ time.
Also, we initialize the data structures $\tmap_\emptyset$ and fixed value $w$,
as described in Theorem~\ref{thm:const-double_strip-online-decision}
in $O(n \log n)$ time,
and insert points in $Q_1$ to have $\tmap_{Q_1}$ in $O(|Q_1|n) = O(n^2)$ time.
We then know that $w^*_{Q_1} \geq w$ or not.

For each $i \geq 2$, there is a point $p\in P$ such that
either $P_i = P_{i-1} \setminus \{p\}$
or $P_i = P_{i-1} \cup \{p\}$.
Since $Q_i = P \setminus P_i$,
we have that either $Q_i = Q_{i-1} \cup \{p\}$
or $Q_i = Q_{i-1} \setminus \{p\}$.
This implies that we can answer whether $w \geq w^*_{Q_i}$ or not
in $O(n)$ time for each $i \geq 2$
by maintaining $P_i$, $Q_i$, and $\tmap_{Q_i}$
by Theorem~\ref{thm:const-double_strip-online-decision}.

Since $m = O(n^2)$ by Lemma~\ref{lem:arrG-zone},
we can solve the decision problem in $O(n^3)$ time.
\begin{lemma} \label{lem:p_annulus_decision}
 Given 
  $w > 0$,
 we can decide whether or not $w \geq w^*$ in $O(n^3)$ time.
\end{lemma}
\begin{proof}
This lemma directly follows Lemma~\ref{lem:const_double-strip_conf}.
Consider any parallelogram annulus $A$ of minimum width $w^*$ that encloses $P$.
Let $D_1$ and $D_2$ be the two double-strips defining $A$.
Also, let $P_1 := P\cap D_1$ and $P_2 := P\cap D_2$.
Then, $D_1$ is a $P$-constrained double-strip enclosing $P_1 \subseteq P$
and $D_2$ is a $P$-constrained double-strip enclosing $P_2 \subseteq P$.
By Lemma~\ref{lem:const_double-strip_conf},
$D_1$ can be replaced by a $P$-constrained double-strip $D'_1$ enclosing $P_1$
of minimum width $w^*_{P_1}$ such that
either (a) three extreme points of $P$ lie on the boundary of its outer strip
or (b) two points of $P_1$ lie on the boundary of its inner strip.
In the same way,
$D_2$ can be replaced by a$ P$-constrained double-strip $D'_2$ enclosing $P_2$
of minimum width $w^*_{P_2}$ such that
either (a) three extreme points of $P$ lie on the boundary of its outer strip
or (b) two points of $P_2$ lie on the boundary of its inner strip.

Let $A'$ be the parallelogram annulus defined by $D'_1$ and $D'_2$.
Observe that $A'$ also encloses $P$, as $P_1 \subset D'_1$ and $P_2 \subset D'_2$.
We further have that $w^* = \max \{w^*_{P_1}, w^*_{P_2}\}$
since $D_1$ and $D_2$ have widths at most $w^*$ and
$w^*$ is the minimum possible width for parallelogram annuli enclosing $P$.

Without loss of generality, we assume that $w^*_{P_1} \geq w^*_{P_2}$,
so $w^*$ is equal to the width $w^*_{P_1}$ of $D'_1$.
Let $\theta \in [-\pi/2, \pi/2)$ be the orientation of $D'_1$.
Then, by Observation~\ref{obs:const_double-strip},
the outer strip of $D'_1$ is $S(\theta)$ and
its inner strip is $S'_{P_1}(\theta)$.
We have two cases:
(a) three extreme points of $P$ lie on the boundary of
the outer strip $S(\theta)$ of $D'_1$, or
(b) two points of $P_1$ lie on the boundary of the inner strip
$S'_{P_1}(\theta)$ of $D'_1$.

In the former case,
the function $d_p$ for every $p\in P$ has a breakpoint at $\theta$.
Let $q \in P_1$ be a point lying on the boundary of $S'_{P_1}(\theta)$.
We then have
 \[ w^* = w^*_{P_1} = \max_{p\in P_1} d_p(\theta) = d_q(\theta).\]
So, $w^*$ is equal to the $y$-coordinate of the vertex of $\arr(\Gamma)$
that corresponds to the breakpoint of $d_q$ at $\theta$.

In the latter case,
we have two points $q_1, q_2\in P_1$ lying on the boundary of $S'_{P_1}(\theta)$.
This implies that
 \[ w^* = w^*_{P_1} = \max_{p\in P_1} d_p(\theta)
    = d_{q_1}(\theta) = d_{q_2}(\theta).\]
Therefore, $w^*$ is equal to the $y$-coordinate of the vertex of $\arr(\Gamma)$
that is the intersection point between a curve in $\Gamma_{q_1}$ and
another in $\Gamma_{q_2}$.

This completes the proof of the lemma.
\end{proof}

Finally, we describe our algorithm to compute the exact value of $w^*$.
To do so, we collect a set $W$ of candidate width values in which $w^*$ is guaranteed to exist, and perform a binary search using our decision algorithm
summarized in Lemma~\ref{lem:p_annulus_decision}.
\begin{lemma} \label{lem:p_annulus_cadidates}
 The minimum possible width $w^*$ over all parallelogram annulli that enclose $P$
 is equal to the $y$-coordinate of a vertex of $\arr(\Gamma)$.
\end{lemma}

We thus define $W$ to be the set of $y$-coordinates of
all vertices of $\arr(\Gamma)$.
Lemma~\ref{lem:p_annulus_cadidates} guarantees that $w^* \in W$.
After sorting the values in $W$,
we perform a binary search on $W$ using the decision algorithm.
Since $|W| = O(n^3)$ by Lemma~\ref{lem:arrG}
and the decision algorithm runs in $O(n^3)$ time by Lemma~\ref{lem:p_annulus_decision},
we can find the exact value of $w^*$ in $O(n^3 \log n)$ time.
Therefore, we conclude the following theorem.
\begin{theorem}
 Given a set $P$ of $n$ points,
 a minimum-with parallelogram annulus over all pairs of orientations
 that encloses $P$ can be computed in $O(n^3 \log n)$ time.
\end{theorem}

%



%

{
\bibliographystyle{abbrv}
\bibliography{annuli}
}

\end{document}